\documentclass[11pt]{amsart}
\usepackage{geometry}                
\usepackage{graphicx}
\usepackage{amssymb}
\usepackage{epstopdf}

\usepackage{tikz}
\usetikzlibrary{matrix}
  \usepackage{paralist}
  \usepackage{epsfig} 
\usepackage{graphicx}  \usepackage{epstopdf}
 \usepackage[colorlinks=true]{hyperref}

\newtheorem{theorem}{Theorem}[section]

\newtheorem{proposition}{Proposition}

\theoremstyle{definition}

\newtheorem{remark}{Remark}

\title{Alternative Angle-Based Approach to the $\mathcal{KS}$-Map. An Interpretation Through Symmetry and Reduction}
\author{S. Ferrer$^\dag$ and F. Crespo$^\ddag$}

\begin{document}
\maketitle


\begin{center}
Space Dynamic Group, DITEC, Facultad Inform\'atica\\
Universidad de Murcia\\
30100 Campus de Espinardo. Murcia, Spain$^\dag$
\end{center}

\vspace{0.5cm}
\begin{center}
Grupo GISDA, Dept. de Matem\'atica, Facultad de Ciencias\\
Universidad del B\'{i}o-B\'{i}o\\
Av. Collao 1202. Concepci\'on, Chile$^\ddag$
\end{center}

\begin{abstract}
The $\mathcal{KS}$ map is revisited in terms of an $S^1$-action in $T^*\mathbb{H}_0$ with the bilinear function as the associated momentum map. Indeed, the $\mathcal{KS}$ transformation maps the $S^1$-fibers related to the mentioned action to single points. By means of this perspective a second twin-bilinear function is obtained with an analogous $S^1$-action. We also show that the connection between the 4-D isotropic harmonic oscillator and the spatial Kepler systems can be done in a straightforward way after regularization and through the extension to 4 degrees of freedom of the Euler angles, when the bilinear relation is imposed. This connection incorporates both bilinear functions among the variables. We will show that an alternative regularization separates the oscillator expressed in Projective Euler variables. This setting takes advantage of the two bilinear functions and another integral of the system including them among a new set of variables that allows to connect the 4-D isotropic harmonic oscillator and the planar Kepler system. In addition, our approach makes transparent that only when we refer to rectilinear solutions, both bilinear relations defining the $\mathcal{KS}$ transformations are needed.
\end{abstract}


\section{Introduction}
The Kustaanheimo-Stiefel transformation, from now on $\mathcal{KS}$, was introduced first in \cite{Kustaanheimo1964} in spinor formulation. However, in \cite{KustaanheimoStiefel1965} the knowledge of spinor is suppressed and the transformation is presented within a matrix setting. Its importance relies in the fact that it connects two of the most remarkable systems of classical mechanics. Namely, the 4-D isotropic oscillator and the spatial Kepler systems, which are some of the very rare few examples of maximally superintegrable and are defined by the parametric Hamiltonian functions
\begin{equation}
\label{O4G}
\mathcal{H}_\omega= \frac{1}{2}\sum_i^4(p_i^2 +\, \omega\,q_i^2),
\end{equation}
and
\begin{equation}
\label{ }
\mathcal{K}_\mu= \frac{1}{2}| \mathbf{y}|^2 -\dfrac{\mu}{| \mathbf{x}|}
\end{equation}
where $\omega$ and $\mu$ are positive  parameters and $\mathbf{x},\mathbf{y}\in T^*\mathbb{R}_0^3$. The isotropic oscillator and the bounded Kepler system describe orbits with the same geometry. However, from the dynamical point of view they show significant differences due to the origin displacement from the center to the foci. Precisely, the incompleteness of the Kepler system flow may be the major discrepancy and boosted a number of regularizing efforts. The connection between both systems is not an evident issue and has a long history that goes back to the time of Euler \cite{Euler1767}, which for the one dimensional Kepler motion already established the connection with the one dimensional harmonic oscillator by introducing a square-root coordinate $q=\sqrt{x}$ and a fictitious time. To mimic the process in the two dimensional case took more than one and a half century, when Levi-Civita \cite{LeviCivita1920} bridged the gap by introducing a conformal map. The final step was carried out by Kustaanheimo and Stiefel in a striking way \cite{KustaanheimoStiefel1965}, since the spatial Kepler system is linked to the four dimensional oscillator. A few years later Moser \cite{Moser1970} also regularized the 3-D Kepler mapping it with the geodesic flow in the four dimensional sphere. 

In our study we focus on the $\mathcal{KS}$ map hinging on quaternions. This map was obtained by imposing a generalization scheme to the Levi-Civita process and working with matrices. As a result, the generalization of the Levi-Civita to the three dimensional case is not easy and the process ends up with an 4-D system. Although, quaternions were first rejected by Stiefel and Scheifele \cite{Sti} as a suitable tool for describing regularization in celestial mechanics, subsequent studies refuted this idea \cite{Vivarelli1985,Deprit94}. For a concise explanation of this fact see \cite{Waldvogel2008} and the references therein.  

One of our aims is to review the issue of the extra dimension trying to seed some light on this matter. Specifically, we refer to the  bilinear relation appearing in the original definition of this transformation and the key role it performs in the immersion of the 3-D Kepler system in four dimensions. Moreover, we will define a second bilinear relation, named as twin-bilinear, which plays an equivalent role. That is to say, they both are the momentum maps of two $S^1$-actions on $T^*\mathbb{H}_0\equiv T^*\mathbb{C}_0^2\equiv T^*\mathbb{R}_0^4$. The $\mathcal{KS}$ map is closely related to the named action since it maps $S^1$-orbits in $T^*\mathbb{H}_0$ to single points. This $S^1$-point association of the $\mathcal{KS}$ map was already evidenced in \cite{Saha2009} from a geometrical interpretation point of view in terms of quaternions.

Additionally, we will show how the Projective Euler variables, which are a customized version of the classical Euler angles studied in \cite{FerrerCrespo2014}, provide a straightforward way for reaching the $\mathcal{KS}$ transformation. Precisely, these set of variables provides local coordinates for $T^*\mathbb{H}_0$ that allows for ``reading" the angles of the $S^1$ orbits associated to the mentioned actions. The actual oscillator-Kepler connection is carried out by changing to the Projective Euler variables plus a regularization. However, we discuss two possible regularizations here; one of them allows for the mentioned connection. The other one makes the oscillator in these variables separable by splitting it in two 1-DOF subsystems. Moreover, the action-angle variables associated to one of these subsystems (the spherical rotor in the classical Euler angles) provides a new set of variables in which the connection of the 4-D oscillator is now established with the planar Kepler system. This connection is now determined regardless of the value of the bilinear and twin-bilinear momentum map. Thus, our approach made transparent that only when we refer to rectilinear solutions, the bilinear relation defining the $\mathcal{KS}$ transformation is needed.

Finally we outline the way in which the paper is organized. In Section~\ref{sec:RegularizingMaps} we briefly review the Levi-Civita and $\mathcal{KS}$ transformations. The next Section~\ref{sec:S3Reduction} introduces the actions allowing to explain the $\mathcal{KS}$ map and the bilinear relations in this context. Section~\ref{sec:Angles} is devoted to the Projective Euler coordinates, which provides a direct derivation of the $\mathcal{KS}$ map by incorporating the angles associated to the action among the new variables. The following step is to incorporate more symmetries to the variables. This process is carried out by the Projective Andoyer coordinates, which connects the 4-D oscillator with a planar Kepler system.

\section{The Levi-Civita and $\mathcal{KS}$ maps}
\label{sec:RegularizingMaps}
Before moving on to the substance of the Levi-Civita and the $\mathcal{KS}$ map we include here some useful notation in quaternions that we will use later. The interested reader would find a nice reference on quaternions in the text of \cite{Kuipers}.

The division ring of quaternions is denoted by $\mathbb{H}$ and is generated by the 4th four roots of unity $\{1,\boldsymbol{i},\boldsymbol{j},\boldsymbol{k}\}$, which satisfy the following rules
$$\boldsymbol{ii}=\boldsymbol{jj}=\boldsymbol{kk}=-1,\quad \boldsymbol{ijk}=-1.$$
Thus we have that $\boldsymbol{ij}=\boldsymbol{jk}=\boldsymbol{-ji}=-1.$ The elements in $\mathbb{H}$ are then expressed as 
$$\mathbf{ x}=x_1+x_2\boldsymbol{i}+x_3\boldsymbol{j}+x_4\boldsymbol{k}.$$
These elements are made of a real part $\mathrm{Re}\left[\mathbf{ x}\right]=x_1$ plus the imaginary vector part $\mathrm{Im}\left[\mathbf{ x}\right]=(x_2\boldsymbol{i}+x_3\boldsymbol{j}+x_4\boldsymbol{k})$, which may be considered as a vector in $\mathbb{R}^3$. Sometimes it will be convenient to adopt the following notation $\mathbf{x}=x_1+x$, where $x=\mathrm{Im}\left[\mathbf{ x}\right]$. Quaternions having zero scalar part are called pure or imaginary and may be regarded both as a vector in $\mathbb{R}^3$ and as a quaternion. Multiplication is performed in the usual manner, like polynomial multiplication, taking the above relations into account or alternatively making use of the ``dot'' and ``cross'' product in $\mathbb{R}^3$
\begin{equation}
\label{QuaternionMultiplication}
\mathbf{x} \cdot \mathbf{y}=(x_1y_1-x\cdot y\,,\,x_1y\,+\,y_1x\,+\,x\times y).
\end{equation}
For the sake of a cleaner notation we will drop the dot. Finally, we define $\mathbf{x}^*=(x_1,-x)$ to be the conjugate quaternion of $\mathbf{x}$. Finally, the norm of $\mathbf{x}$ is the length as a four dimensional vector and is denoted by $|\mathbf{x}|=\sqrt{\mathbf{x}\,\mathbf{x}^*}=\sqrt{\mathbf{x}^*\mathbf{x}}$. 

\subsection{The Levi-Civita transformation}
\label{sec:LC}
Along this section we will identify $\mathbb{R}^2$ with $\mathbb{C}$ and the complex numbers as immersed in the quaternions by considering the last two coefficient corresponding to $\boldsymbol{j}$ and $\boldsymbol{k}$ identically zero. Hence, let $\mathbf{ x},\mathbf{ q}\in\mathbb{C}$ be given by $\mathbf{ x}={ x}_1+{ x}_2\boldsymbol{i}$ and $\mathbf{ q}={ q}_1+{ q}_2\boldsymbol{i}$, the 2-D Levi-Civita map reads as follows
\begin{equation}
\label{eq:LC2d}
\mathbf{ x}=\mathbf{ q}\cdot 1 \cdot \mathbf{ q},\qquad { x}_1={ q}_1^2-{ q}_2^2,\; { x}_2=2{ q}_1{ q}_2.
\end{equation}
This transformation may be lifted to the cotangent bundle $T^*\mathbb{C}_0$ by imposing $\mathbf{ p}d\mathbf{ q}=\mathbf{ y}d\mathbf{ x}$. Indeed, there are three more alternatives to (\ref{eq:LC2d}) leading to equivalent transformations, all of them correspond to any of the following ones
\begin{gather}
\begin{aligned}
\label{eq:LC2dAlternatives}
&\mathbf{ x}=\mathbf{ q}\cdot -1 \cdot \mathbf{ q},&\qquad &{ x}_1={ q}_2^2-{ q}_1^2,\; { x}_2=-2{ q}_1{ q}_2,\\
&\mathbf{ x}=\mathbf{ q}\cdot \boldsymbol{i} \cdot \mathbf{ q},&\qquad &{ x}_1=-2 { q}_1 { q}_2,\; { x}_2={ q}_1^2 - { q}_2^2,\\
&\mathbf{ x}=\mathbf{ q}\cdot -\boldsymbol{i} \cdot \mathbf{ q},&\qquad &{ x}_1=2 { q}_1 { q}_2,\; { x}_2={ q}_2^2 - { q}_1^2.
\end{aligned}
\end{gather}
Thus, the definition of the Levi-Civita transformation
$$\mathcal{LC}:T^*\mathbb{C}_0\longrightarrow T^*\mathbb{C}_0,\quad (\mathbf{ q},\mathbf{ p})\rightarrow(\mathbf{ x},\mathbf{ y})$$
may be carried out with any of the above coordinate mappings. Here we recall the expression in the first case
\begin{gather}
\begin{aligned}
\label{eq:LC2dCompleta}
&{ x}_1={ q}_1^2-{ q}_2^2,\; { x}_2=2{ q}_1{ q}_2\\
&{ y}_1=\frac{q_1p_1-q_2p_2}{2\,({ q}_1^2+{ q}_2^2)},\; { y}_2=\frac{q_1p_2+q_2p_1}{2\,({ q}_1^2+{ q}_2^2)}.
\end{aligned}
\end{gather}

The Levi-Civita transformation may be used for regularization purposes. It is well known that the flow of the Kepler system is not complete. Hence, to avoid this drawback in the planar case, this transformation was introduced in \cite{LeviCivita1920}. Indeed, let us consider the Hamiltonian of the Kepler system
$$ \mathcal{K}:T^*\mathbb{R}_{0}^{2} \longrightarrow \mathbb{R};\quad \mathbf{(x,y)}\rightarrow\frac{1}{2}|\mathbf{y}|^2-\frac{\mu}{|\mathbf{x}|}, $$
where $\mu>0$. In addition, for $h\neq0$ we also examine the auxiliary Hamiltonian 
 \begin{eqnarray}\label{fhamil}
 \mathbf{H}(\mathbf{ x,y})=\frac{|\mathbf{x} |}{h}(\mathcal{K}+2h^2 )+\frac{\mu}{h},
 \end{eqnarray}
which associated Hamiltonian vector field read as follows
 \begin{eqnarray} \label{hamil}
\dot{\mathbf{x}}&=&\frac{|\mathbf{x}|}{h}\,\frac{\partial\mathcal{K}}{\partial \mathbf{y}},  \nonumber \\
\dot{\mathbf{y}}&=&-\frac{|\mathbf{x}|}{h}\,\frac{\partial\mathcal{K}}{\partial \mathbf{x}}-\left[\frac{\partial}{\partial \mathbf{x}} \left(\frac{|\mathbf{x}|}{h}\right)(\mathcal{K}+2 h^2 ) \right].
\end{eqnarray}
Note that for the energy level $\mathcal{K}=-2h^2$ the vector field associated with $ \mathcal{H}$ matches with $ \mathcal{K}$ after applying the time regularization $dt={|\mathbf{x}|}/{h}ds$. Then, after a straightforward computation (\ref{eq:HamAux}) may be given by
\begin{equation}\label{eq:HamAux}
 \mathbf{H}\mathbf{(x,y)}=\frac{|\mathbf{x}|}{E}\left[\frac{1}{2}|\mathbf{y}|^2+2 h^2 \right],
 \end{equation}
which after making use of the Levi-Civita transformation becomes
\begin{equation}
\mathcal{H}\mathbf{(q,p)}=\frac{1}{2}\left(\frac{|\mathbf{p}|^2}{4h}+4h|\mathbf{q}|^2\right).\nonumber
\end{equation}
Finally, by assuming $\mathbf{q}={\mathbf{q}}/({2\sqrt{h}})$ and $\mathbf{p}=2\mathbf{p}\sqrt{h}$ and with a slight abuse of notation, we get to the following more convenient expression
\begin{equation}
\mathcal{H }(\mathbf{q,p})=\frac{1}{2}\left(|\mathbf{q}|^2+|\mathbf{p}|^2\right).
\end{equation}
In other words, the Levi-Civita map and the time scaling connected the 2-D Kepler and the 2-D isotropic oscillator for a fixed energy value.

\subsection{The Kustaanheimo-Stiefel map}
\label{sec:KS}
The situation changes when we face the three dimensional Kepler system since the generalization of the 2-D case is not straightforward. The celebrated paper of \cite{KustaanheimoStiefel1965} represents a cornerstone in the regularization of the spatial Kepler system, it was an attempt of extending the Levi-Civita map to a transformation mapping two 4-dimensional spaces. Although they did not succeed in defining such a transformation, these authors came with a mapping of $\mathbb{R}^4$ onto $\mathbb{R}^3$, the well known $\mathcal{KS}$ map, which was good enough for regularization purposes. 

In this section we recall some established facts regarding the $\mathcal{KS}$ map. Then, they will be reconsidered from a new point of view. We start with the definition of the $\mathcal{KS}$ map, which may be considered as a Hopf map type. In the literature there is not a common agreement in the exact expression of this map as well as the way in which it is introduced. Namely, the definition may be given by using an orthogonal matrix, \cite{KustaanheimoStiefel1965,Roa2016}, Pauli matrices, \cite{Cushman97,vdMeer2015} or alternatively hinging in quaternions \cite{Saha2009}. We will follow the quaternionic approach defining the $\mathcal{KS}$ map as it is given in \cite{Saha2009}
\begin{gather}
\begin{aligned}
\label{eq:KSmap}
&\mathcal{KS}:T^* \mathbb{H}_0\longrightarrow \Sigma\subset T^* \mathbb{H}_0\quad (\mathbf{ q},\mathbf{ p})\rightarrow(\mathbf{ x},\mathbf{y}),\\
& \mathbf{ x}=\mathbf{ q}^*\,\boldsymbol{k}\,\mathbf{ q}, \quad \mathbf{ y}=\dfrac{\mathbf{ q}^*\,\boldsymbol{k}\,\mathbf{ p}}{2\mathbf{ q}^*\,\mathbf{ q}},
\end{aligned}
\end{gather}
where  $T^* \mathbb{K}_0^n\equiv\mathbb{K}^n-\{0\}\times\mathbb{K}^n$ and the explicit formulas are given by
\begin{gather}
\begin{aligned}
\label{eq:KSmapCoor}
& \mathbf{ x}=\left(0,2\,( q_2 q_4 -q_1 q_3),2\, (q_1 q_2 +  q_3 q_4), q_1^2 - q_2^2 - q_3^2 + q_4^2\right),\\
& \mathbf{ y}=\dfrac{1}{2\mathbf{  q}^*\,\mathbf{  q}}\big(\,p_1 q_4-p_4 q_1+ p_2 q_3 - p_3 q_2 ,\\
&\phantom{ \mathbf{ y}=\dfrac{1}{2\mathbf{  q}^*\,\mathbf{  q}}(\,}  p_4 q_2   + p_2 q_4- p_1 q_3-p_3 q_1 ,\\
&\phantom{ \mathbf{ y}=\dfrac{1}{2\mathbf{  q}^*\,\mathbf{  q}}(\,} p_2 q_1 + p_1 q_2 + p_4 q_3 + p_3 q_4,\\
&\phantom{ \mathbf{ y}=\dfrac{1}{2\mathbf{  q}^*\,\mathbf{  q}}(\,}  p_1 q_1 - p_2 q_2 - p_3 q_3 + p_4 q_4\big).
\end{aligned}
\end{gather}
Considering the restriction of the $\mathcal{KS}$ map to the submanifold $\Xi_0=0$, where we have
\begin{equation}
\label{eq:XiMomentumMapPsi}
\Xi_0(\mathbf{ q},\mathbf{ p})=p_1 q_4-p_4 q_1+ p_2 q_3 - p_3 q_2 .
\end{equation} 
This constraint is equivalent up to indices permutation to the well known bilinear relation and by impossing it, the $\mathcal{KS}$-image set $\Sigma$ becomes $T^* \mathbb{R}_0^3$. Therefore, $\Xi_0$ will be named as the twin-bilinear function and $\Xi_0=0$ the twin-bilinear relation, which yields a surjective Poisson mapping on the image with respect to the Poisson structures comming from the restriction of the standard symplectic form $\Omega=-d \omega$ and $\omega=\mathrm{Re}\left[ \mathbf{ p}^*d\mathbf{ q} \right]$. The original bilinear relation appearing in \cite{KustaanheimoStiefel1965} is given by
\begin{equation}
\label{eq:XiMomentumMapPhi}
\Xi_1(\mathbf{ q},\mathbf{ p})=p_1 q_4-p_4 q_1 + p_3 q_2 - p_2 q_3.
\end{equation}
The reader would appreciate how quaternions help us providing an elegant and compact definition of the $\mathcal{KS}$ map. In addition, it allows for several variations as in the Levi-Civita case. Indeed, the following alternatives lead to a variety of index permutations
\begin{gather}
\begin{aligned}
\label{eq:KSmapAlt}
& \mathbf{ x}=\mathbf{ q}^*\,(\pm\boldsymbol{i})\,\mathbf{ q}, \quad \mathbf{ y}=\dfrac{\mathbf{ q}^*\,(\pm\boldsymbol{i})\,\mathbf{ p}}{2\mathbf{ q}^*\,\mathbf{ q}},\\
& \mathbf{ x}=\mathbf{ q}^*\,(\pm\boldsymbol{j})\,\mathbf{ q}, \quad \mathbf{ y}=\dfrac{\mathbf{ q}^*\,(\pm\boldsymbol{j})\,\mathbf{ p}}{2\mathbf{ q}^*\,\mathbf{ q}},\\
& \mathbf{ x}=\mathbf{ q}^*\,(\pm\boldsymbol{k})\,\mathbf{ q}, \quad \mathbf{ y}=\dfrac{\mathbf{ q}^*\,(-\boldsymbol{k})\,\mathbf{ p}}{2\mathbf{ q}^*\,\mathbf{ q}},\\
\end{aligned}
\end{gather}
Note that this freedom of choice is an instance of the \textit{defining vector} concept introduced in \cite{Breiter2017}. However, the differences between the definitions of the $\mathcal{KS}$ map in \cite{KustaanheimoStiefel1965} and \cite{Saha2009} are not explained by the defining vector\footnote{The definitions of the $\mathcal{KS}$ in \cite{KustaanheimoStiefel1965} and in \cite{Saha2009} differ formally in a permutation of the indices 1 and 4 of the quaternion components.}. Indeed, they require the geometric interpretation of both definitions to be explained. In  \cite{Saha2009} it is pointed out  that given  $\left(\mathbf{ q}_{\beta},\mathbf{ p}_{\beta}\right)$ defined by
\begin{equation}
\label{eq:InterpretGeo1}
\left(\mathbf{ q}_{\beta},\mathbf{ p}_{\beta}\right)=\left(\left(\cos\beta-\sin\beta\,\boldsymbol{k}\right)\mathbf{ q},\left(\cos\beta-\sin\beta\,\boldsymbol{k}\right)\mathbf{ p}\right)
\end{equation}
$\left(\mathbf{ q}_{\beta},\mathbf{ p}_{\beta}\right)$ and $\left(\mathbf{ q},\mathbf{ p}\right)$ correspond with the same image by $\mathcal{KS}$. That is to say, $\mathcal{KS}$ maps $T^* \mathbb{H}_0$ on $ T^* \mathbb{R}_0^3$ and the fiber associated to each $\mathbf{ x}\in\mathbb{R}_0^3$ is $S^1$. A similar, but not equivalent, interpretation may be given to the definition given in \cite{KustaanheimoStiefel1965}. Regarding $T^* \mathbb{H}_0$ as $T^* \mathbb{C}^2_0$, let $\mathbf{ x}\in \mathbb{H}_0$ given by $\mathbf{ x}=(\mathbf{ z},\mathbf{ w})$, where $\mathbf{ z},\mathbf{ w}\in\mathbb{C}$ and we have that $\mathbf{ z}$ corresponds to the indices $1,4$ and $\mathbf{ w}$ is associated to indices $2,3$ of $\mathbf{ x}$. That is to say, we consider $\mathbf{ x}=z_1+w_1\mathbf{ i}+w_2\mathbf{ j}+z_2\mathbf{ k}$. Then, we define $\left(\mathbf{ q}^{\beta},\mathbf{ p}^{\beta}\right)$ by
\begin{equation}
\label{eq:InterpretGeo2}
\left(\mathbf{ q}^{\beta},\mathbf{ p}^{\beta}\right)=\left(e^{i\beta}\mathbf{ u},e^{-i\beta}\mathbf{ v},e^{i\beta}\mathbf{ z},e^{-i\beta}\mathbf{ w}\right).
\end{equation}

\section{Symmetries and Reduction on $T^* \mathbb{H}_0$}
\label{sec:S3Reduction}
The main feature of the $\mathcal{KS}$ map is that it allows to connect the harmonic oscillator and the Kepler system. In the previous Section~\ref{sec:RegularizingMaps} we have seen that $\mathcal{KS}$ hinges in a particular value of the bilinear functions. Now we pay more attention to the remaining bilinear choices and additional quadratic functions that will play a fundamental role. They are defined by
\begin{gather}
\begin{aligned}
\label{eq:Invariants}
 \tau_1(\mathbf{ q},\mathbf{ p})&=q_1p_1+q_2p_2+q_3p_3+q_4p_4 ,  \\
 \tau_2(\mathbf{ q},\mathbf{ p})& = 1/2( |\mathbf{ q}|^2 - |\mathbf{ p}|^2) \; ,  \\
 \tau_3(\mathbf{ q},\mathbf{ p}) &= 1/2 (|\mathbf{ q}|^2 + |\mathbf{ p}|^2) \; ,\\
\rho_1(\mathbf{ q},\mathbf{ p})&=p_2 q_1 - p_1 q_2+p_4 q_3-p_3 q_4 ,\\
\rho_2(\mathbf{ q},\mathbf{ p})&=p_3 q_1-p_1 q_3  + p_4 q_2- p_2 q_4 ,\\
\rho_3(\mathbf{ q},\mathbf{ p})&=\Xi_1(\mathbf{ q},\mathbf{ p})=p_1 q_4-p_4 q_1 + p_3 q_2 - p_2 q_3, \\
\sigma_1 (\mathbf{ q},\mathbf{ p}) &= p_1 q_2-p_2 q_1 + p_3 q_4 - p_4 q_3 ,\\
\sigma_2 (\mathbf{ q},\mathbf{ p}) &=  p_1 q_3-p_3 q_1 + p_4 q_2 - p_2 q_4,\\
\sigma_3 (\mathbf{ q},\mathbf{ p}) &=\Xi_0(\mathbf{ q},\mathbf{ p})=p_1 q_4-p_4 q_1 + p_2 q_3 - p_3 q_2,
\end{aligned}
\end{gather}
for which the following claim was stated in \cite{FerrerCrespo2014}. In $T^*\mathbb{H}_0$ with the standard Poisson bracket $\{,\}$, the  sets of functions 
$$\tau=\{\tau_1,\tau_2,\tau_3\},\quad \rho=\{\rho_1,\rho_2,\rho_3\}, \quad \sigma=\{\sigma_1,\sigma_2,\sigma_3\}$$ 
commute between each other and span Lie algebras in $C^\infty(\mathbb{R}^8)$ isomorphic to $\mathfrak{sl}(2,\mathbb{R})$, $\mathfrak{so}(3) $ and $\mathfrak{so}(3) $ respectively. Therefore, $\mathcal{R} \equiv \mathcal{S} \equiv  SO(3)$ and $ \mathcal{T} \equiv SL(2,\mathbb{R})$, where $ \mathcal{T}$, $ \mathcal{R}$ and $ \mathcal{S}$ are the corresponding Lie groups to the Lie algebras $\tau$, $\rho$ and $\sigma$. In addition the function
\begin{equation}
\label{ }
M(\mathbf{ q},\mathbf{ p})=\frac{1}{2}\sqrt{|\mathbf{ q}|^2|\mathbf{ p}|^2-\langle \mathbf{ q},\mathbf{ p}\rangle^2}
\end{equation}
is the centralizer of $ \mathcal{R}$, $ \mathcal{S}$ and $ \mathcal{T}$. 

Thinking in the connection of the oscillator and the Kepler systems we consider $\tau_3$, $\rho$ and $\sigma$. The first function is the harmonic oscillator and the remaining ones span a Lie algebra isomorphic to $\mathfrak{so}(3)\times \mathfrak{so}(3) \cong \mathfrak{so}(4)$. Moreover, $\rho$ and $\sigma$ are full of ``bilinear functions", that is, $\mathcal{KS}$ map admits alternative definitions by interchanging $\boldsymbol{k}$ in (\ref{eq:KSmap}) with $\boldsymbol{i}$ or $\boldsymbol{j}$, which lead to the alternative bilinear pairs $(\rho_1,\sigma_1)$ and $(\rho_2,\sigma_2)$ respectively. More precisely, $\rho_i$ is always related to a $\mathcal{KS}$ definition equivalent to the one given in \cite{KustaanheimoStiefel1965} and $\sigma_i$ has to do with one definition among (\ref{eq:KSmapAlt}). In what follows we are going to show how besides of the classical bilinear function there are two more that come into play; the twin-bilinear and the centralizer $M$.

The geometric interpretations of the $\mathcal{KS}$ transformations given by (\ref{eq:InterpretGeo1}) and (\ref{eq:InterpretGeo2}) induce in a natural way a pair of $S^1$-actions for which each orbit is mapped to a single point by the $\mathcal{KS}$ map. Namely, let us consider 
\begin{equation}
\label{eq:chiAction}
\chi_i:S^1\times T^*\mathbb{H}_0\longrightarrow T^*\mathbb{H}_0 \qquad \chi_i(\alpha,\mathbf{ q},\mathbf{ p})=(R_{\alpha}^i\mathbf{ q},R_{\alpha}^i\mathbf{ p}),
\end{equation}
being $\mathcal{R}_{\alpha}^0$ and  $\mathcal{R}_{\alpha}^1$ the rotations given by the following matrices 
\begin{equation}
\label{eq:ActionMatrix1}
\mathcal{R}_{\alpha}^0=\left[
\begin{array}{cccc}
 \cos \alpha  & 0 & 0 & -\sin \alpha  \\
 0 & \cos \alpha  & -\sin \alpha  & 0 \\
 0 & \sin \alpha  & \cos \alpha  & 0 \\
 \sin \alpha  & 0 & 0 & \cos \alpha  \\
\end{array}
\right]
\end{equation}
and
\begin{equation}
\label{eq:ActionMatrix2}
\mathcal{R}_{\alpha}^1=\left[
\begin{array}{cccc}
 \cos \alpha  & 0 & 0 & -\sin \alpha  \\
 0 & \cos \alpha  & \sin \alpha  & 0 \\
 0 & -\sin \alpha  & \cos \alpha  & 0 \\
 \sin \alpha  & 0 & 0 & \cos \alpha  \\
\end{array}
\right].
\end{equation}
Both actions are symplectic with respect to the symplectic form $\Omega$ and their momentum maps are given by the bilinear and twin-bilinear functions respectively
\begin{equation}
\label{eq:BilinearFunction}
\Xi_i: T^* \mathbb{H}_0\longrightarrow \mathbb{R}\qquad (\mathbf{ q},\mathbf{ p})\rightarrow \Xi_i(\mathbf{ q},\mathbf{ p}).
\end{equation}

With the above definitions, it is easy to check that the restriction of the $\mathcal{KS}$ map to any orbit of $\Xi_1$, when defined as in \cite{KustaanheimoStiefel1965}, is constant and the same happens when we restrict the definition given in \cite{Saha2009} to any orbit of $\Xi_0$. Although from now on we focus in the definition given in (\ref{eq:KSmap}), what remains of this section is valid to the $\mathcal{KS}$ map no matter what definition you choose. 

Now we specialize to the case in which the bilinear and twin-bilinear relation are imposed, then (\ref{eq:KSmap}) restricts to 
\begin{equation}
\label{eq:KSmapXi}
\mathcal{KS}_0:\mathcal{J}_{0}^i\subset T^* \mathbb{H}_0\longrightarrow T^* \mathbb{R}^3_0\quad (\mathbf{ q},\mathbf{ p})\rightarrow(\mathbf{ x},\mathbf{y}),
\end{equation}
where $\mathcal{J}_{0}^i=\Xi_i^{-1}(0)$.
%
%
\begin{proposition}
The $\mathcal{KS}$ map induces a symplectomorphism between the regular $\chi_i$-reduced space $\mathcal{J}_{0}^i/S^1$ and  $T^* \mathbb{R}^3_0$.
\end{proposition}
\begin{proof}
Let us consider the canonical projection on the $\chi_i$-reduced space
$$\pi:\mathcal{J}_{0}^i\longrightarrow \mathcal{J}_{0}^i/S^1,$$
Given $\mathbf{ z}=(\mathbf{ q},\mathbf{ p})\in\mathcal{J}_{0}^i$, we denote by $\mathcal{O}_{\mathbf{ z}}\subset \mathcal{J}_{0}^i$ the $\chi_i$-orbit through $\mathbf{ z}$. Thus, $\pi(\mathbf{ w})=\left[\mathbf{ z}\right]$ for all $\mathbf{ w}\in\mathcal{O}_{\mathbf{ z}}$. That is, the projection of $\mathcal{O}_{\mathbf{ z}}$ in the reduced space is the class $\left[\mathbf{ z}\right]$. Then, we have that the following map
$$\widetilde{\mathcal{KS}}_{0}:\mathcal{J}_{0}^i/S^1\longrightarrow T^* \mathbb{R}_0^3,\quad \left[\mathbf{ z}\right]\rightarrow \mathcal{KS}_{0}(\mathbf{ w})$$
where $\mathbf{ w}$ is any point in $\mathcal{O}_{\mathbf{ z}}$. Note that $\mathcal{KS}$ transformation maps orbits into single points which guarantees that $\widetilde{\mathcal{KS}}_{0}$ is well defined. Moreover, this map is a bijection and the following diagram is commutative 

\begin{tikzpicture}
  \matrix (m) [matrix of math nodes,row sep=3em,column sep=4em,minimum width=2em]
  {
     \mathcal{J}_{0}^i &  T^* \mathbb{R}_0^3 \\
     \mathcal{J}_{0}^i/S^1 &  \\};
  \path[-stealth]
    (m-1-1) edge node [left] {$\pi$} (m-2-1)
    (m-1-1) edge node [above] {$\mathcal{KS}_{0}$} (m-1-2)
    (m-2-1) edge node [below] { $\phantom{xxx}\widetilde{\mathcal{KS}}_{0}$} (m-1-2);
\end{tikzpicture}
\par\noindent

Thus, since $\pi$ is a canonical map and so does $\mathcal{KS}_{0}$, we have that $\widetilde{\mathcal{KS}}_{0}$ is a symplectomorphism between $ \mathcal{J}_{0}^i$ and $T^* \mathbb{R}_0^3$.
\end{proof}

Traditionally in the literature \cite{KustaanheimoStiefel1965}, the bilinear relation has been emphasized versus the twin-bilinear relation. However, it is clear that both are equivalent.

\section{The $\mathcal{KS}$ Map Through the Projective Euler Variables}
\label{sec:Angles}
In this section we will provide a different way of deriving the $\mathcal{KS}$ map based in a modified version of the Euler parameters. This set of variables is a modified version of the one considered by \cite{Ferrer2010} and we will show that the angles of the $S^1$ reductions given by the actions (\ref{eq:chiAction}) are among the Projective Euler variables, which are defined by means of the following transformation
$\mathcal{PE}_F: (\rho,\phi,\theta,\psi,P,\Phi,\Theta,\Psi)\rightarrow (\mathbf{ q},\mathbf{ p})$
\begin{eqnarray}\label{parametros}
\vspace{-0.3cm}
&&\hspace{-0.8cm}q_1=\sqrt{\rho}\, \cos\frac{\theta}{2} \sin\frac{\phi+\psi}{2},\,
q_3=\sqrt{\rho}\, \sin\frac{\theta}{2} \sin\frac{\phi-\psi}{2}, \\
&&\hspace{-0.85cm}q_2=\sqrt{\rho}\, \sin\frac{\theta}{2} \cos\frac{\phi-\psi}{2},\,
q_4=\sqrt{\rho}\,\cos\frac{\theta}{2} \cos\frac{\phi+\psi}{2},\nonumber
\end{eqnarray}
with $(\rho,\phi,\theta,\psi)\in R^+\times(0,2\pi)\times (0,\pi)\times
(0,2\pi)$. This set of variables is not new, see for instance \cite{Ikeda1970,Sti,Bar,Cor}). 
The associated momenta are obtained by canonical extension of (\ref{parametros}), that is,  $\sum p_idq_i= P\,d\rho+\Phi\, d\phi+\Theta\, d\theta+\Psi\,d\psi$
\begin{eqnarray}\label{momentos}
&& \hspace{-0.3cm}P = \frac{1}{2\sum q_i^2}(q_1p_1+q_2p_2+q_3p_3+q_4p_4),\nonumber\\[0.8ex]
&&\hspace{-0.3cm}\Theta=\frac{(q_1p_4+q_4p_1)(q_3^2 +q_2^2) -(q_2p_3+q_3 p_2)(q_1^2 +q_4^2)}
{2\,\sqrt{(q_1^2 + q_4^2)(q_2^2+q_3^2)}},\nonumber\\
&&\hspace{-0.3cm}\Phi=\frac{1}{2}\,\Xi_1=\frac{1}{2}(p_1 q_4-p_4 q_1 - p_2 q_3 + p_3 q_2),\\
&&\hspace{-0.3cm}\Psi=\frac{1}{2}\,\Xi_0 =\frac{1}{2}(p_1 q_4-p_4 q_1+ p_2 q_3 - p_3 q_2 ).\nonumber
\end{eqnarray}
A complete treatment of this variables can be found in \cite{FerrerCrespo2014,Crespo2015}. Note that these variables endow both, the $\hat{\chi}$ and ${\chi}_i$ reduced spaces with a suitable set of local coordinates. 


Then, excluding the invariant manifolds $\mathcal{M}_1\!\!=\{(q,Q)|q_1\!=\!q_4\!=\!0\}$ and $\mathcal{M}_2=\{(q,Q)|q_2\!=\!q_3\!=\!0\}$ where Levi-Civita transformation already shows the connection of the 2-D  Kepler and isotropic oscillators,  the Hamiltonian (\ref{O4G}) in the new variables may be written as
\begin{eqnarray}\label{newham}
&&\hspace{-0.8cm}\mathcal{H}_\omega=\mathcal{H}(\rho,\theta,-,-,P,\Theta,\Phi,\Psi)\nonumber\\
&&\hspace{-0.2cm}=  \frac{\rho\, \omega}{2} +2\rho P^2 +
\frac{2}{\rho}\;\left(\Theta^2  +
\frac{\Phi^2+\Psi^2-2\,\Phi\Psi\,\cos\theta}{\sin^2\theta} \right)
\end{eqnarray}
{\it i.e.} variables $\phi$ and $\psi$ are cyclic, with
$\Phi$ and $\Psi$ as the corresponding first integrals. Thus, the system of differential equations reduces to
\[\frac{d\rho}{d\tau}=\frac{\partial \mathcal{H}}{\partial P},\quad 
\frac{d\theta}{d\tau}=\frac{\partial \mathcal{H}}{\partial \Theta},\quad
\frac{dP}{d\tau}=-\frac{\partial \mathcal{H}}{\partial \rho},\quad
\frac{d\Theta}{d\tau}=-\frac{\partial \mathcal{H}}{\partial \theta}\]
and two quadratures 
\begin{equation}\label{doscuadraturas}
\phi=\int (\partial \mathcal{H}/\partial \Phi)\,d\tau \quad {\rm and} \quad 
\psi=\int(\partial \mathcal{H}/\partial \Psi)\,d\tau.
\end{equation}

\subsection{On the Role of the Regularizations}
\label{sec:Regularizations}
Although the 4-D harmonic oscillator in Cartesian variables is already linear, regular and separable, the preceding change of variables could be interpreted as an unnecessary effort. However, here is when regularizations come into scene to show the  4-D oscillator's dynamical richness.
\begin{proposition}
\label{propo:Separability}
After regularization, the 4-D harmonic oscillator expressed in Projective Euler variables becomes separable and includes the dynamics of the spheric rotor.
\end{proposition}
\begin{proof}
Considering the Hamiltonian (\ref{newham}) we fix the energy level $\mathcal{H}_\omega=h$ and  according to Poincar\'e technique with $d\tau=(\rho/4)ds$, the Hamiltonian becomes
\begin{equation}
\label{eq:HamSeparado}
\mathcal{K}=\frac{\rho}{4}(\mathcal{H}_\omega-h)=\mathcal{K}_{\rho}+\mathcal{K}_{\theta}
\end{equation}
where
\begin{equation}
\label{eq:HamSeparadoBis}
\mathcal{K}_{\rho}=\frac{\omega\rho^2}{8}+ \frac{\rho^2P^2}{2}-\frac{h\,\rho}{4} ,\quad \mathcal{K}_{\theta}= \frac{1}{2}\left(\Theta^2  +
\frac{\Phi^2+\Psi^2-2\,\Phi\Psi\,\cos\theta}{\sin^2\theta} \right),
\end{equation}
in the manifold $\mathcal{K}=0$. That is to say, $\mathcal{K}$ has been separated in two 1-DOF subsystems being $\mathcal{K}_{\theta}$ the spherical rotor.  
\end{proof}

Note that the Hamiltonian $\mathcal{K}_{\theta}$ is obtained by making equal all the principal moments of inertia of the free rigid body system in Euler angles, see \cite{MarsdenRatiu}.

The main result of this work is to show the relationship of the $\mathcal{KS}$ map and the Projective Euler variables. For this purpose just a  rearrangement of (\ref{newham}) and a regularization are needed to reach the spatial Kepler system. The next result is also given in a preliminary version of this work by one of the authors \cite{Ferrer2010}.

\begin{theorem}
\label{theorem:Euler}
The restriction of the 4-D harmonic oscillator to either $\Phi=0$ or $\Psi=0$ in Projective Euler variables becomes a 3-D Kepler system and a quadrature.
\end{theorem}
\begin{proof}
After a  time regularization $\tau\rightarrow s$  given  by 
\begin{equation}\label{regular1}
d\tau=(4\rho)^{-1}\,ds.
\end{equation}
the flow is described by the Hamiltonian $\tilde{\mathcal{H}}=\frac{1}{4\rho}(\mathcal{H}- h)$ on the manifold $\tilde{\mathcal{H}}=0$, where $h$ is a fixed value $\mathcal{H}=h$. Equivalently, we can replace $\tilde{\mathcal{H}}$ by
\begin{equation}\label{ham3}
\tilde{\mathcal{K}}=\frac{1}{2}\left(P^2 + \frac{\Theta^2}{\rho^2}+\frac{\Phi^2
+\Psi^2-2\,\Phi\Psi\,\cos\theta}{\rho^2\,\sin^2\theta}
\right)- \frac{h}{4\rho}
\end{equation}
in the manifold $\tilde{\mathcal{K}}=-\frac{\omega}{8}$. Moreover, the following rearrangement is convenient
\begin{equation}\label{hamkeplerbis}
\tilde{\mathcal{K}}= \mathcal{H}_K + \frac{\Psi^2-2\Phi\Psi\cos \theta}{2\rho^2\,\sin^2\theta},
\end{equation}
where denoting $\gamma= h/4$, we have that $\mathcal{H}_K$ is given by
\begin{equation}\label{KeplerCoulomb}
\mathcal{H}_K=\frac{1}{2}\left(P^2 + \frac{\Theta^2}{\rho^2}+\frac{\Phi^2}{\rho^2\,\sin^2\theta}
\right)- \frac{\gamma}{\rho}.
\end{equation}
That is to say, the Hamiltonian of the Kepler system in spherical coordinates, see \cite{Goldstein}
\begin{gather}
\begin{aligned}
\label{esfericascarte}
x&=\rho\,\sin\theta\,\cos\phi,\\
y&=\rho\,\sin\theta\,\sin\phi,\\
z&=\rho\,\cos\theta,\quad \\
X&=\frac{1}{2\,\rho\,\sin\theta } \left[\cos\phi \left(\Theta  \sin2 \theta+2 P \rho  \sin ^2\theta \right)-2 \Phi  \sin\phi\right]&\\
Y&=\frac{1}{2\,\rho\,\sin\theta }  \left[\sin\phi \left(\Theta  \sin(2 \theta )+2 P \rho  \sin ^2\theta \right)+2 \Phi  \cos\phi\right]&\\
Z&=P \cos \theta-\frac{\Theta  \sin\theta}{\rho }&
\end{aligned}
\end{gather}

Be aware of the fact that we are free to choose between the terms ${\Psi^2}/({\rho^2\sin^2\theta})$ or
${\Phi^2}/({\rho^2\sin^2\theta})$ for the definition of $\mathcal{H}_K$. In other words, we have to equivalent bilinear relations.

The system of differential equations defined by (\ref{hamkeplerbis}) reads as follows
\begin{eqnarray}\label{ecuaciones}
&&\frac{d\rho}{ds}=\phantom{-}\frac{\partial \tilde{\mathcal{K}}}{\partial P}=\phantom{-}
\frac{\partial \mathcal{H}_K}{\partial P},\nonumber\\[0.5ex]
&&\frac{d\theta}{ds}=\phantom{-}\frac{\partial \tilde{\mathcal{K}}}{\partial \Theta}=\phantom{-}\frac{\partial \mathcal{H}_K}{\partial \Theta},\\[0.5ex]
&&\frac{dP}{ds}=-\frac{\partial \tilde{\mathcal{K}}}{\partial \rho}=-\frac{\partial \mathcal{H}_K}{\partial \rho}+\frac{\Psi(\Psi-2\Phi\,\cos \theta)}{\rho^3\,\sin^2\theta},\nonumber\\ [0.5ex]
&&\frac{d\Theta}{ds}=-\frac{\partial \tilde{\mathcal{K}}}{\partial \theta}=-\frac{\partial \mathcal{H}_K}{\partial \theta}- \frac{\Psi(\Psi\cos \theta-\Phi(1+\cos^2 \theta))}{\rho^2\sin^3\theta},\nonumber
\end{eqnarray}
together with the quadratures (\ref{doscuadraturas}). Considering our choice in the definition of (\ref{KeplerCoulomb}), we restrict to the manifold $\Psi=0$. Thus, the harmonic oscillator (\ref{ecuaciones}) with the 
quadrature $ \phi=\int (\partial \mathcal{H}_K/\partial \Phi)\,ds$, reduces to the Keplerian system in
the manifold $\mathcal{H}_K=-\frac{\omega}{8}$. Once the system (\ref{ecuaciones}) is integrated we obtain $\psi$ 
$$\psi=\int \frac{\partial \tilde{\mathcal{K}}}{\partial \Psi}\,ds=  -\int \frac{\Phi \cos \theta}{\rho^2\sin^2\theta} \,ds , $$
by injecting the expressions for $\Phi$, $\theta$ and $\rho$.
\end{proof}

Theorem~\ref{theorem:Euler} relates the 4-D harmonic oscillator with a generalized Kepler system. Indeed, the Hamiltonian (\ref{hamkeplerbis}) is connected to Hartmann and other ring-shaped potentials \cite{Kib87}.   

Note that, considering the inverse of the Projective Euler transformation (\ref{parametros}), we may see the transformation 
from spherical to Cartesian (\ref{esfericascarte}) as a projection $\mathbb{H}\rightarrow \mathbb{R}^3$. More precisely, reversing (\ref{parametros}) we have $\rho = \sum q_i^2$ and 
\begin{eqnarray}\label{Eulerinversa}
&&\sin\theta =\frac{2\Delta}{\rho}, \quad \cos\theta =\frac{q_1^2 -q_2^2-
q_3^2 +q_4^2 }{\rho},\nonumber\\
&&\sin\phi = \frac{q_1q_2+q_3q_4}{\Delta},
\quad \cos\phi = \frac{q_2q_4-q_1q_3}{\Delta},\nonumber\\
&&\sin\psi = \frac{q_1q_2-q_3q_4}{\Delta},
\quad \cos\psi = \frac{q_2q_4+q_1q_3}{\Delta},\nonumber
\end{eqnarray}
where $\Delta=\sqrt{(q_1^2 +q_4^2)(q_2^2 +q_3^2)}$ and taking into account that 
$$\Psi =\dfrac{1}{2}\,\Xi_0=\frac{1}{2}(p_1 q_4-p_4 q_1 - p_3 q_2 + p_2 q_3)=0$$
we obtain 
\begin{gather}
\begin{aligned}\label{eq:MomentosAndoyer}
&x=2(q_2q_4-q_1q_3),  \;                   &  X  =\dfrac{1}{2\rho}\,( p_4 q_2   + p_2 q_4- p_1 q_3-p_3 q_1),  \\
&y=2(q_1q_2+q_3q_4),   \;                 &  Y  =\dfrac{1}{2\rho}\,( p_2 q_1 + p_1 q_2 + p_4 q_3 + p_3 q_4), \\
&z= q_1^2 -q_2^2-q_3^2+q_4^2,     &  Z  =\dfrac{1}{2\rho}\, (p_1 q_1 - p_2 q_2 - p_3 q_3 + p_4 q_4),
\end{aligned}
\end{gather}
in other words, the $\mathcal{KS}$ map, which may also be obtained as the transformation making commutative the diagram in Figure~\ref{fig:DiagramaEuler}.
\begin{figure}[h]
\centering
\begin{tikzpicture}
  \matrix (m) [matrix of math nodes,row sep=3em,column sep=4em,minimum width=2em]
  {
     T^* \mathbb{H}_0 & T^* \mathbb{H}_0 \\
     T^* \mathbb{R}^3_0 & T^* \mathbb{R}^3_0 \\};
  \path[-stealth]
    (m-1-1) edge node [left] {$\mathcal{KS}$} (m-2-1)
    (m-1-1) edge node [below] {$\mathcal{PE}_F$} (m-1-2)
    (m-2-2) edge node [below] {$\Gamma$} (m-2-1)
    (m-1-2) edge node [right] {$\pi$} (m-2-2);
\end{tikzpicture} 
\caption{Commutative diagram. The map $\Gamma$ is the transformation from spherical to Cartesian coordinates and $\pi$ is the projection $(\rho,\phi,\theta,\psi,R,\Phi,\Theta,\Psi)\rightarrow (\rho,\phi,\theta,R,\Phi,\Theta)$.}
\label{fig:DiagramaEuler}
\end{figure}

Up to now we have presented the $\mathcal{KS}$ map as the connection between the flow of the 4-D oscillator and the bounded spatial Kepler system. However, both systems are super-integrable and hence their flows are made of periodic closed curves, which in addition define planar trajectories. Thus, one could ask for a connection beyond the $\mathcal{KS}$ map by moving to the Keplerian orbital plane. That is to say, we look for a connection of the 4-D oscillator and the planar Kepler system. This task is carried out by choosing a suitable set of variables, which are given by a partial transformation of the phase space. Indeed, these new coordinates are obtained by taking into account the separation of variables given in Proposition~\ref{propo:Separability}, where the pair $(\rho,\,P)$ is fixed and the remaining variables $(\phi,\,\theta,\,\psi,\,\Phi,\,\Theta,\,\Psi)$ are transformed to the action-angle variables for the Hamiltonian $\mathcal{K}_{\theta}$. Details are to be found in \cite{Heard}, pag. 88. Gathering of all of them,  it leads to the Projective Andoyer variables given in  \cite{FerrerCrespo2014,Crespo2015}
$$\mathcal{PA}_C: (\rho,\lambda,\mu,\nu,P,\Lambda,M,N)\rightarrow (\mathbf{ q},\mathbf{ p}),$$
which, after some computations and using quaternionic notation, lead to the following compact formula for the explicit change of variables
\begin{gather}
\begin{aligned}
\label{eq:AndoyerVariablesMomenta}
&  \mathbf{q}=\sqrt{\rho}\,\mathbf{q}_{\mathbf{k},\nu}\,\mathbf{q}_{\mathbf{i},J}\,\mathbf{q}_{\mathbf{k},\mu}\,\mathbf{q}_{\mathbf{i},I}\,\mathbf{q}_{\mathbf{k},\lambda},\\
& \mathbf{p}=\dfrac{1}{4\rho^2}\,\mathbf{q} \,\mathbf{w}^*,
\end{aligned} 
\end{gather}
where $\mathbf{q}_{\mathbf{l},\alpha}=(\cos\alpha+\mathbf{l}\sin\alpha)$ for $\mathbf{l}=\mathbf{i},\mathbf{j},\mathbf{k}$ and  $\mathbf{w}=(\tau_1, \rho_1, \rho_2, \rho_3)$, which components were defined in (\ref{eq:Invariants}) and are given by
\begin{gather}
\begin{aligned}
\label{eq:OmegasAndoyerGsAndoyer}
& \tau_1= \rho P,  \\ \nonumber
& \rho_1=2\sqrt{M^2-N^2}\sin \nu, \\ 
& \rho_2=2\sqrt{M^2-N^2}\cos \nu,\\ 
&\rho_3=  2N.
\end{aligned} 
\end{gather}

Note that these variables are not defined for $M=0$, since the motion is rectilinear. The following result shows how the bilinear and twin-bilinear functions can be used to connect the 4-D harmonic oscillator with a planar Kepler system. 
\begin{theorem}
\label{therorem:Andoyer}
System (\ref{O4G}) in Projective Andoyer variables, when properly regularized,
includes the planar Keplerian system for any value of the bilinear functions $\Lambda$ and $N$.
\end{theorem}
\begin{proof}
Let us consider the regularization $g(\rho) = 1/(4\rho)$ and the Projective Andoyer variables. Then, considering $\gamma=h/4$ we obtain
\begin{equation}\label{ESAT}
\tilde{\mathcal{K}} = g(\rho) (\mathcal{H}_{\omega}-h)= \frac{1}{2}\Big(P^2 + \frac{M^2}{\rho^2} \Big) - \frac{\gamma}{\rho}
\end{equation}
in the manifold $ \tilde{\mathcal{K}} = -  \omega/8$, which is the hamiltonian of the planar Kepler system in polar coordinates, see \cite{Meyer}  
\end{proof}

\begin{remark}
It is important to be aware of the fact that in our development $(\lambda,\nu,\Lambda, N)$ are integrals for which no particular value need to be fixed at any time. This feature is in contrast with the fixed value of the bilinear relation identically zero. The explanation of this apparent contradiction is that the cases $\Lambda\neq0$ or $N\neq0$ do not allow for $M=0$, which correspond with rectilinear orbits. Thus, Projective Andoyer variables connect the 4-D oscillator and the non-rectilinear 2-D Kepler system.
\end{remark}
\begin{figure}[h]
\centering
\begin{tikzpicture}
  \matrix (m) [matrix of math nodes,row sep=3em,column sep=4em,minimum width=2em]
  {
     T^* \mathbb{H}_0 & T^* \mathbb{H}_0 \\
     T^* \mathbb{R}^2_0 & T^* \mathbb{R}^2_0 \\};
  \path[-stealth]
    (m-1-1) edge node [left] {$\Upsilon$} (m-2-1)
    (m-1-1) edge node [below] {$\mathcal{PA} 	_F$} (m-1-2)
    (m-2-2) edge node [below] {$\Sigma$} (m-2-1)
    (m-1-2) edge node [right] {$\pi$} (m-2-2);
\end{tikzpicture} 
\caption{Commutative diagram. The map $\Sigma$ is the transformation from polar to Cartesian coordinates and $\pi$ is the projection $(\rho,\lambda,\mu,\nu,R,\Lambda,M,N)\rightarrow (\rho,\mu,R,M)$.}
\label{fig:DiagramaAndoyer}
\end{figure}

The diagram given in Figure~\ref{fig:DiagramaAndoyer} is the Andoyer analogue version to the corresponding diagram in Figure~\ref{fig:DiagramaEuler}. In the Euler diagram the $\mathcal{KS}$ map makes it commutative, for the Andoyer case it is a matter of study to look for a $\Upsilon$ map making the Andoyer diagram commutative.

\section{Conclusions}
\label{sec:Conclusions}
Besides of the valuable geometric insight of the $\mathcal{KS}$ map given in \cite{Saha2009}, it also admits an interpretation in terms of the $S^1$-reductions associated to the bilinear or the twin-bilinear functions. Moreover, it is also shown the way in which the old angles coming from celestial mechanics are nothing but the angles associated to the fibre of the mentioned  actions. This feature of the angles coordinates enables them for a very simple and straightforward derivation of the $\mathcal{KS}$ map. Precisely, by means of the Projective Euler variables we show in a direct way the connection between the 4-D isotropic harmonic oscillator and the bounded spatial Kepler system, providing an alternative construction of the $\mathcal{KS}$ map. In addition, we have shown that there are more symmetries playing a role in this setting, \textit{i.e.}, the twin-bilinear relation and the centralizer. Projective Andoyer variables take advantage of these symmetries to connect the 4-D isotropic harmonic oscillator and the bounded planar Kepler system. As a coming task, these authors would like to investigate the existence, properties and expression in terms of quaternions of $\Upsilon$, the analogous to the $\mathcal{KS}$ map given in Figure~\ref{fig:DiagramaAndoyer}.



\section*{Acknowledgements}

Support from Research Agencies of Spain and Chile is acknowledged. They came in the form of research projects MTM2015-64095-P and ESP2013-41634-P, of the Ministry of Science of Spain and from the project 11160224 of the Chilean national agency FONDECYT. Thanks to J. Zapata for helping us checking some expressions. 





\begin{thebibliography}{}

\bibitem[Barut et~al., 1979]{Bar}
Barut, A., Schneider, C., and Wilson, R. (1979).
\newblock Quantum theory of infinite quantum theory of infinite component
  fields.
\newblock {\em J. Math. Phys.}, 20.

\bibitem[Breiter and K., 2017]{Breiter2017}
Breiter, S. and K., L. (2017).
\newblock Kustaanheimo-stiefel transformation with an arbitrary defining
  vector.
\newblock {\em Celest. Mech. Dynamical Astron.}, (128):323--342.

\bibitem[Cornish, 1984]{Cor}
Cornish, F. H.~J. (1984).
\newblock The hydrogen atom and the four-dimensional harmonic oscillator.
\newblock {\em Journal of Physics A: Mathematical and General}, 17(2):323.

\bibitem[Crespo, 2015]{Crespo2015}
Crespo, F. (2015).
\newblock Hopf fibration reduction of a quartic model. {A}n application to
  rotational and orbital dynamics.
\newblock {\em PhD. Universidad de Murcia}.

\bibitem[Cushman and Bates, 2015]{Cushman97}
Cushman, R. and Bates, L. (2015).
\newblock {\em Global Aspects of Classical Integrable Systems}.
\newblock Birkh\"auser Verlag, Basel, 2nd edition.

\bibitem[Deprit et~al., 1994]{Deprit94}
Deprit, A., Elipe, A., and Ferrer, S. (1994).
\newblock Linearization: Laplace vs. stiefel.
\newblock {\em Celestial Mechanics and Dynamical Astronomy}, 58(2):151--201.

\bibitem[Euler, 1767]{Euler1767}
Euler, L. (1767).
\newblock De motu rectilineo trium corporum se mutuo attrahentium.
\newblock {\em Novi Commentarii academiae scientiarum Petropolitanae},
  11:144--151.

\bibitem[Ferrer, 2010]{Ferrer2010}
Ferrer, S. (2010).
\newblock The projective {A}ndoyer transformation and the connection between
  the 4-d isotropic oscillator and {K}epler systems.
\newblock {\em arXiv:1011.3000v1 [nlin.SI]}.

\bibitem[Ferrer and Crespo, 2014]{FerrerCrespo2014}
Ferrer, S. and Crespo, F. (2014).
\newblock Parametric quartic {H}amiltonian model. {A} unified treatment of
  classic integrable systems.
\newblock {\em Journal of Geometric Mechanics}, 6(4):479--502.

\bibitem[Goldstein et~al., 2002]{Goldstein}
Goldstein, H., Poole, C., and Safko, J. (2002).
\newblock {\em Classical Mechanics}.
\newblock Addison Wesley, New York, Third edition.

\bibitem[Heard, 2006]{Heard}
Heard, W. (2006).
\newblock {\em Rigid Body Mechanics}.
\newblock WILEY-VCH Verlag GmbH \& Co. KGaA, Mathematics, Physics and
  Applications.

\bibitem[Ikeda and Miyachi, 1971]{Ikeda1970}
Ikeda, M. and Miyachi, Y. (1971).
\newblock On the mathematical structure of the symmetry of some simple
  dynamical systems.
\newblock {\em Matematica Japoniae}.

\bibitem[Kibler and Winternitz, 1987]{Kib87}
Kibler, M. and Winternitz, P. (1987).
\newblock Dynamical invariance algebra of the hartmann potential.
\newblock {\em Journal of Physics A: Mathematical and General}, 20(13):4097.

\bibitem[Kuipers, 1999]{Kuipers}
Kuipers, J. (1999).
\newblock {\em Quaternions and rotation sequences}.
\newblock Princeton University text, Princeton, New Jersey.

\bibitem[Kustaanheimo, 1964]{Kustaanheimo1964}
Kustaanheimo, P. (1964).
\newblock Spinor regularization of the kepler motion.
\newblock {\em Annales Universitatis Turkuensis}, 73(3).

\bibitem[Kustaanheimo and Stiefel, 1965]{KustaanheimoStiefel1965}
Kustaanheimo, P. and Stiefel, E. (1965).
\newblock Perturbation theory of {K}epler motion based on spinor
  regularization.
\newblock {\em J. Reine Angew. Math.}, 218:204--219.

\bibitem[Levi-Civita, 1920]{LeviCivita1920}
Levi-Civita, T. (1920).
\newblock Sur la r\'egularisation du probl\`eme des trois corps.
\newblock {\em Acta Mathematica}, 42.

\bibitem[Marsden and Ratiu, 1999]{MarsdenRatiu}
Marsden, J.~E. and Ratiu, T.~S. (1999).
\newblock {\em Introduction to \textsc{M}echanics and \textsc{S}ymmetry}.
\newblock Springer-Verlag New York, Inc., 2nd edition.

\bibitem[Meyer et~al., 2009]{Meyer}
Meyer, K., Hall, G.~R., and Offin, D. (2009).
\newblock {\em Introduction to Hamiltonian Dynamical Systems and the N-Body
  Problem, 2nd Ed}, volume~90 of {\em APM}.
\newblock Applied Mathematical Sciences, Springer, New York.

\bibitem[Moser, 1970]{Moser1970}
Moser, J. (1970).
\newblock Regularization of {K}epler's problem and the averaging method on a
  manifold.
\newblock {\em Communications on pure and applied mathematics}, XXIII:609--636.

\bibitem[Roa et~al., 2016]{Roa2016}
Roa, J., Urrutxua, H., and Pel{\'a}ez, J. (2016).
\newblock Stability and chaos in kustaanheimo--stiefel space induced by the
  hopf fibration.
\newblock {\em Monthly Notices of the Royal Astronomical Society},
  459(3):2444--2454.

\bibitem[Saha, 2009]{Saha2009}
Saha, P. (2009).
\newblock Interpreting the {K}ustaanheimo-{S}tiefel transform in gravitational
  dynamics.
\newblock {\em Mon. Not. R. Astron. Soc.}, 400:228--231.

\bibitem[Stiefel and Scheifele, 1971]{Sti}
Stiefel, E. and Scheifele, G. (1971).
\newblock Linear and regular celestial mechanics.
\newblock {\em Springer, Berlin}.

\bibitem[van~der Meer, 2015]{vdMeer2015}
van~der Meer, J. (2015).
\newblock The kepler system as a reduced 4d harmonic oscillator.
\newblock {\em Journal of Geometry and Physics}, 92(Supplement C):181 -- 193.

\bibitem[Vivarelli, 1985]{Vivarelli1985}
Vivarelli, M.~D. (1985).
\newblock The ks-transformation in hypercomplex form and the quantization of
  the negative-energy orbit manifold of the kepler problem.
\newblock {\em Celestial mechanics}, 36(4):349--364.

\bibitem[Waldvogel, 2008]{Waldvogel2008}
Waldvogel, J. (2008).
\newblock Quaternions for regularizing celestial mechanics: the right way.
\newblock {\em Celest. Mech. Dynamical Astron.}, 102:149--162.

\end{thebibliography}
\end{document}